\documentclass[a4]{article}

\usepackage[affil-it]{authblk}

\usepackage[margin={2cm,2cm}]{geometry}
\usepackage{amssymb}
\usepackage{amsmath}
\usepackage{amsthm}
\usepackage{xfrac}
\usepackage{graphicx}
\usepackage{hyperref}
\usepackage{bbold}
\usepackage{multicol}

\newtheorem{theorem}{Theorem}
\newtheorem{proposition}[theorem]{Proposition}
\newtheorem{corollary}[theorem]{Corollary}
\newtheorem{example}[theorem]{Example}
\newtheorem{definition}[theorem]{Definition}

\def\M{\mathcal{M}}
\def\E{\mathcal{E}}
\newcommand{\ket}[1]{\left| #1 \right>}

\begin{document}

\title{On the `Reality' of Observable Properties}
\author{Shane Mansfield
\thanks{shane.mansfield@cs.ox.ac.uk; The author thanks Samson Abramsky, Clare Horsman, Nadish de Silva and Rui Soares Barbosa for comments and discussions.}
}
\affil{Quantum Group, Department of Computer Science, \\ University of Oxford}
\date{\today}

\maketitle

\begin{abstract}
This note contains some initial work on attempting to bring recent developments in the foundations of quantum mechanics concerning the nature of the wavefunction within the scope of more logical and structural methods. A first step involves generalising and reformulating a criterion for the reality of the wavefunction proposed by Harrigan \& Spekkens, which was central to the PBR theorem. The resulting criterion has several advantages, including the avoidance of certain technical difficulties relating to sets of measure zero. By considering the `reality' not of the wavefunction but of the observable properties of any ontological physical theory a new characterisation of non-locality and contextuality is found. Secondly, a careful analysis of preparation independence, one of the key assumptions of the PBR theorem, leads to an analogy with Bell locality, and thence to a proposal to weaken it to an assumption of `no-preparation-signalling' in analogy with no-signalling. This amounts to introducing non-local correlations in the joint ontic state, which is, at least, consistent with the Bell and Kochen-Specker theorems. The question of whether the PBR result can be strengthened to hold under this relaxed assumption is therefore posed.
\end{abstract}

\begin{multicols}{2}

\section{Introduction}

The issue of the reality of the wavefunction has received a lot of attention recently (see especially \cite{pusey:12,colbeck:12}). In this note, we show that insights may also be gained by taking a similar approach to considering the `reality' of objects and properties in physical theories more generally, and in particular that such an approach can provide a new perspective on non-locality and contextuality.
The first step will be to introduce a suitably general criterion for `reality' inspired by the Harrigan-Spekkens criterion for the reality of the wavefunction \cite{harrigan:10}, which was the subject of the Pusey-Barrett-Rudolph theorem \cite{pusey:12}.

The aim is to formulate the ideas in a manner that can allow for a deeper, structural understanding of what is at play. Indeed, the initial motivation was to bring considerations of this kind within the scope of the methods of the unified sheaf-theoretic approach to non-locality and contextuality \cite{abramsky:11,abramsky:11a,mansfield:13b}. The resulting, more general criterion has several advantages. It avoids certain technical difficulties, and due to its generality it can be applied within any ontological physical theory: generalised probabilistic theories \cite{barrett:07}, or classical mechanics, for example. 

The initial investigations here also demonstrate that such considerations can provide an alternative perspective on foundational questions in general. We find an alternative characterisation of both local and non-contextual correlations, as those that can arise from observations of properties that are `real' in this sense. This ties together the notions of locality and reality, bringing to light another link between the Bell and Pusey-Barrett-Rudolph (PBR) theorems \cite{pusey:12}, which deal, respectively, with these properties.

We begin, in section \ref{ontic}, by presenting our generalisation and reformulation of the criterion for reality, which requires minimal background. Much of the foundations of quantum mechanics literature, including that concerning recent developments on the reality of the wavefunction, deals with hidden variable or ontological models. Therefore, before considering some uses of the more general criterion in section \ref{observables}, we will provide a brief review of this framework in section \ref{ontsect}. It has already been pointed out that local hidden variable models can be subsumed by the sheaf-theoretic framework \cite{brandenburger:11,brandenburger:13}; we will arrive at another proof of this fact in section \ref{observables}.

Finally, in section \ref{sec:pbr}, we provide some comments on the property of preparation independence, which first appeared as one of the assumptions of the PBR theorem. We show that the assumption, which is crucial to the theorem, is analogous to Bell locality in a precise sense; yet another link between the Bell and PBR theorems. We propose a weakening of this assumption to one analogous to no-signalling \cite{ghirardi:80}, which we believe can still be well-motivated. It is not clear, however, if it is possible to strengthen the PBR theorem so that its result still holds under the less strict assumption. Certainly, such a strengthening would require a different argument. We will also mention that, by assuming preparation independence, one can very easily prove a weakened form of Bell's theorem \cite{bell:64}, a fact that may cast suspicion on the strength of the stricter assumption.

\section{A Criterion for Reality}\label{ontic}

In this section we will use the terminology of Harrigan \& Spekkens \cite{harrigan:10}, which has been established in the literature. We begin by reviewing their criterion for the reality, or \emph{onticity}, of the wavefunction, which will then be reformulated and generalised. For this, we need only postulate, for each system, a space $\Lambda$ of \emph{ontic states}. These can be considered to correspond to real, physical states of the system. The idea will be that objects or properties that are determined with certainty by the ontic state are themselves ontic. The term ontic is meant to describe that which relates to real as opposed to phenomenal existence, though we do not propose to get into a discussion of the suitability of the terminology here. Similarly, objects or properties that are not determined with certainty are said to be \emph{epistemic}; the literal meaning of the term is that which relates to knowledge or to its degree of validation. The use of the term in \cite{harrigan:10} can be taken to reflect the fact that objects and properties of this kind are necessarily probabilistic and could thus be assumed to represent a degree of knowledge about some underlying ontic object or property. It should be borne in mind, of course, that results relating to these definitions will hold regardless of the physical significance attached to them.

As well as the existence of an ontic state space, the authors of \cite{harrigan:10} also posit the assumption that the preparation of any quantum state $\left| \psi \right>$ induces a distribution $\mu_{\left| \psi \right>}$ over the ontic state space $\Lambda$ for that system, specifying the probabilities for the system to be in each ontic state given that it has been prepared in this way.

\begin{definition}[Harrigan \& Spekkens \cite{harrigan:10}]\label{def:hs}
If, for all wavefunctions $\left| \psi \right> \neq \left| \phi \right>$ of any system, the induced distributions $\mu_{\left| \psi \right>}$ and $\mu_{\left| \phi \right>}$ have non-overlapping supports, the wavefunction is said to be \emph{ontic}. Otherwise, there exist some $\left| \psi \right> \neq \left| \phi \right>$ such that $\mu_{\left| \psi \right>}(\lambda)>0$ and $\mu_{\left| \phi \right>}(\lambda)>0$ for some $\lambda \in \Lambda$, and the wavefunction is said to be \emph{epistemic}.
\end{definition}

We now present a more general reformulation of the definition. As we will see, this can be applied to any object or property. Though the wavefunction would more usually be considered to be (at least) a mathematical object rather than a property of a system, for simplicity we only refer to properties from now on.

\begin{definition}\label{def:ontepi}
A \emph{$\mathcal{V}$-valued property} over $\Lambda$ is a function $f: \Lambda \rightarrow \mathcal{D}(\mathcal{V})$, where $\mathcal{D}(\mathcal{V})$ is the set of probability distributions over $\mathcal{V}$.
The property is said to be \emph{ontic} in the special case that, for all $\lambda \in \Lambda$, the distribution $f(\lambda)$ over $\mathcal{V}$ is a delta function.
Otherwise, it is said to be \emph{epistemic}.
\end{definition}

Another way of stating this is that ontic properties are generated by functions $\widehat{f}:\Lambda \rightarrow \mathcal{V}$; i.e. they map each ontic state to a unique value. For epistemic properties, however, there is at least one ontic state that is compatible with two or more distinct values in $\mathcal{V}$.

We now set about showing how these definitions relate, which may not be immediatelyz\ clear. Any $\mathcal{V}$-valued property $f$ specifies probability distributions over $\mathcal{V}$, conditioned on $\Lambda$. Bayesian inversion can be used to obtain probability distributions over $\Lambda$, conditioned on $\mathcal{V}$, which we (suggestively) label $\{ \mu_v \}_{v \in \mathcal{V}}$. Explicitly,
\begin{equation}\label{distdef}
\mu_v(\lambda) := \frac{ f(\lambda) (v) \cdot p(\lambda)}{ \int_{\Lambda} f(\lambda') (v) \cdot p(\lambda') \, d\lambda'},
\end{equation}
assuming a uniform distribution $p(\lambda)$ on $\Lambda$. Note that this is only well-defined for finite $\Lambda$, and that a more careful measure theoretic treatment, which will not be provided here, is required for the infinite case.

\begin{proposition}\label{prop:hscor}
A $\mathcal{V}$-valued property over finite $\Lambda$ is ontic (according to definition \ref{def:ontepi}) if and only if the distributions $\{ \mu_v \}_{v \in \mathcal{V}}$ have non-overlapping supports.
\end{proposition}

\begin{proof}
Suppose the property $f$ is ontic in the sense of definition \ref{def:ontepi}, let $\lambda \in \Lambda$, and let $v,v' \in \mathcal{V}$ such that $v \neq v'$. Assume for a contradiction that $\mu_v(\lambda)>0$ and $\mu_{v'}(\lambda)>0$. Then, by (\ref{distdef}), $f(\lambda) (v)>0$ and $f(\lambda) (v')>0$; but since $f$ is ontic,
\[
v_\lambda = v \neq v' = v_\lambda,
\]
where $v_\lambda := \widehat{f}(\lambda)$.

Conversely, suppose that the distributions $\{ \mu_v \}_{v \in \mathcal{V}}$ have non-overlapping supports and assume for a contradiction that $f(\lambda) (v)>0$ and $f(\lambda) (v')>0$. Then, by (\ref{distdef}), $\mu_v(\lambda)>0$ and $\mu_{v'}(\lambda)>0$.
\end{proof}

One way of thinking about this correspondence is as a special case of the dual equivalence between the category of von Neumann algebras and $*$-homomorphisms and the category of measure spaces and measurable functions \cite{heunen:13}.

To illustrate, we provide a couple of simple examples of ontic and epistemic properties.

\begin{example}[Classical Mechanics]
The phase space of a system is taken to be the ontic state space. Classical mechanical observables (energy, momentum, etc.) are represented by real-valued functions on phase space, and are therefore ontic.
\end{example}

\begin{example}[Fuzzy Measurement]\label{ex:epi}
Consider an experiment in which a bag is prepared containing two coins, which can be green or white, with equal probability, but are otherwise identical. We claim that the process of removing one and checking its colour measures an epistemic property. If the ontic states are $\Lambda = \{ GG, GW, WG, WW \}$, the property cannot be represented by a $\{G,W\}$-valued function on $\Lambda$. Given the ontic state $GW$, for example, both $G$ and $W$ are compatible, and can arise with equal probability.
\end{example}

In this second example, the property that is being measured is, according to the present definition, epistemic with respect to the state of the bag; it might also be said the example describes a fuzzy measurement on the state of the bag.

Definition \ref{def:ontepi} has several advantages.
\begin{itemize}
\item
It is fully general and can be applied to any object or property in any ontological theory.
\item
It avoids measure theoretic problems relating to sets of measure zero that are inherent to the original.
\item
It is mathematically straightforward and conceptually transparent.
\end{itemize}

\section{Ontological Models}\label{ontsect}

We are concerned with theories that give operational predictions for outcomes to measurements; we refer to sets of such predictions as empirical models. Quantum mechanics is one such theory, which might be described operationally by saying that we associate a density matrix $\rho^p$ with each preparation $p$, a POVM $\{E^m_o\}_{o \in O}$ with each measurement $m$, and prescribe the probability of the outcome $o$ given preparation $p$ and measurement $m$ by
\[
p(o\mid m,p) = \text{tr}(\rho^p \, E^m_o).
\]

We wish, more generally, to consider theories with the same kind of operational structure. In order to do so, we will use some notation that is similar to that of the sheaf-theoretic approach. For each system we assume spaces $P$ of preparations, $X$ of measurements, and $O$ of outcomes. There may be some compatibility structure on the space of measurements, say $\mathcal{M} \subseteq \mathcal{P}(X)$, specifying which sets of measurements can be made jointly (in quantum mechanics, this is specified by the commutative sub-algebras of the algebra of observables). This information encodes which kind of measurement scenario we are working in: e.g. the Bell-CHSH model \cite{clauser:69,bell:87}, Hardy model \cite{hardy:92,hardy:93}, and PR correlations \cite{popescu:94} all deal with two-party scenarios in which each party can choose freely between two binary-outcome measurements\footnote{This measurement scenario is referred to as the $(2,2,2)$ Bell scenario.}. Again, we additionally assume a space $\Lambda$ of \emph{ontic states}, over which each preparation induces a probability distribution.

In an effort to simplify notation, we will use an overline to denote a joint measurement $\overline{m} \in \mathcal{M}$ or joint outcomes $\overline{o} \in \mathcal{E}(\overline{m})$ to that measurement; here $\mathcal{E}(\overline{m})$ is the set of functions $\overline{o}:\overline{m} \rightarrow O$. Readers familiar with the sheaf-theoretic approach will recall that $\mathcal{E}:X \rightarrow O^X$ is the event sheaf. On the other hand, $m\in X$ and $o \in O$ denote individual measurements and outcomes, respectively. Joint preparations and joint ontic states will be treated similarly in section \ref{sec:pbr}.

\begin{definition}\label{def:hv}\index{hidden variable model}
An \emph{ontological} or \emph{hidden variable model} $h$ over $\Lambda$ specifies:
\begin{enumerate}
\item
A distribution $h(\lambda \mid p)$ over the ontic states $\Lambda$ for each preparation $p \in P$;
\item
For each ontic state $\lambda \in \Lambda$ and joint measurement $\overline{m} \in \mathcal{M}$, a distribution
\begin{equation}\label{eq:ontstatestats}
h(\overline{o}\mid \overline{m},\lambda)
\end{equation}
over joint outcomes $\E(\overline{m})$.
\end{enumerate}
The \emph{operational probabilities}\index{operational probabilities} are then prescribed by
\begin{equation}\label{hv}
h(\overline{o}\mid \overline{m},p) = \int_\Lambda d \lambda \; h(\overline{o}\mid\overline{m},\lambda) \; h(\lambda\mid p).
\end{equation}
\end{definition}

The terms ontological model and hidden variable model are both used in the literature, but recently the term ontological model has gained some popularity. It may be a more suitable term in the sense that the `hidden' variable need not necessarily be hidden at all: it could be directly observable. In Bohmian mechanics \cite{bohm:52,bohm:52a}, for example, position and momentum play the role of the hidden variable. It also carries the connotation that such a model is an attempt to describe some underlying ontological reality.

\begin{definition}
A theory which determines the measurement statistics for the ontic states (\ref{eq:ontstatestats}) will be referred to as an \emph{ontological theory} over $\Lambda$.
\end{definition}

We are especially interested in ontological models and theories that can reproduce quantum mechanical predictions. Trivially, the simplest such theory is quantum mechanics itself, regarded as an ontological theory.

\begin{example}[$\psi$-complete Quantum Mechanics]\index{$\psi$-completeness}
The ontic state is identified with the quantum state. A preparation produces a density matrix, which is regarded as a distribution over the projective Hilbert space associated with the system. By construction, the operational probabilities are those given by the Born rule.
\end{example}

Of course, quantum mechanics, treated as an ontological theory in itself in this way, has certain non-intuitive features. Einstein, Podolsky \& Rosen provided one early discussion of the fact \cite{einstein:35}; but later results such as Bell's theorem \cite{bell:64} and the Kochen-Specker theorem \cite{kochen:75} clarified the fact that non-locality and contextuality are necessary features of any theory that can account for quantum mechanical predictions. In order to address these issues, we point out some relevant properties that ontological models may have.

\begin{definition}\index{$\lambda$-independence}
An ontological model is \emph{$\lambda$-independent} if and only if the distributions over $\Lambda$ induced by each preparation $p \in P$ do not depend on the joint measurement $\overline{m} \in \mathcal{M}$ to be performed.
\end{definition}

We have already implicitly assumed this in definition \ref{def:hv}, but it is worth making it clear since it is a crucial assumption in all of the familiar no-go theorems. In a $\lambda$\emph{-dependent} model, on the other hand, the probabilities of being in the various ontic states depend on both the preparation of the system and the joint measurement being performed, and we would have $h(\lambda\mid p,\overline{m})$ rather than $h(\lambda\mid p)$ in equation (\ref{hv}).

\begin{definition}
An ontological model is \emph{deterministic} if and only if for each $\lambda \in \Lambda$ and set of compatible measurements $\overline{m} \in \mathcal{M}$ there exists some joint outcome $\overline{o} \in \mathcal{E}(\overline{m})$ such that $h(\overline{o}\mid \overline{m},\lambda)=1$.
\end{definition}

In such a model, the outcome to any measurement that can be performed on an ontic state is determined with certainty.

For any distribution $h(\overline{o}\mid \overline{m},\lambda)$ over joint outcomes $\overline{o} \in \mathcal{E}(\overline{m})$ to the joint measurement $\overline{m} \in \M$ on the ontic state $\lambda \in \Lambda$, we can find a distribution $h(o\mid m, \lambda)$ over outcomes $o \in O$ to any individual measurement $m \in \overline{m}$ by marginalisation. 

\begin{definition}
An ontological model is \emph{parameter-independent} if and only if the probability distribution $h(o\mid m, \lambda)$ over $O$ is well-defined for each $m \in X$ and $\lambda \in \Lambda$.
\end{definition}

By well-definedness we mean that the same marginal distribution $h(o\mid m, \lambda)$ is obtained regardless of which set of joint of measurements we marginalise from (in the case that $m \in \overline{m}$ and $m \in \overline{m}'$ for example). Parameter independence thus asserts that the probabilities of outcomes to a particular measurement do not depend on the other measurements being performed. Essentially, it imposes no-signalling \cite{ghirardi:80} with respect to the ontic states.

\begin{definition}\index{locality \& non-contextuality!ontological model/theory}
An ontological model is \emph{local/non-contextual} if and only if it is both deterministic and parameter-independent; empirical correlations are \emph{local (non-contextual)} if and only if they can be realised by a local (non-contextual) model.
\end{definition}

This says that for each ontic state there is a certain outcome to any measurement that can be performed, and that this does not depend on which other measurements are made. The term local is generally only used when the system being modelled is spatially distributed; where such an arrangement is not assumed, the model is said to be non-contextual.

We draw attention to the fact that another definition of locality that is common in the literature concerns the factorisability of the distributions
\begin{equation}\label{eq:bellloc}
h(\overline{o}\mid \overline{m},\lambda) = \prod_{m \in \overline{m}} \, h(o\mid m,\lambda).
\end{equation}
While the present definition may be less familiar, it is important to note that these definitions were shown to be equivalent, in the sense that they generate the same sets of empirical models, in \cite{abramsky:11}, which built on work by Fine \cite{fine:82} that was specific to the $(2,2,2)$ Bell scenario.

\section{Observable Properties}\label{observables}

If we assume that the outcomes of measurements provide the values of properties of a system, then for each measurement $m \in X$ there should exist an $O$-valued property $f_m: \Lambda \rightarrow \mathcal{D}(O)$ such that $f_m(\lambda) (o) = h(o\mid m,\lambda)$ for all $\lambda \in \Lambda$ and $o \in O$.
\begin{definition}\label{def:obsprop}
The \emph{observable properties} of an ontological model $h$ over $\Lambda$ are the $O$-valued properties $f_m: \Lambda \rightarrow \mathcal{D}(O)$ given by
\begin{equation}\label{prophv}
f_m(\lambda) (o) := h(o\mid m,\lambda)
\end{equation}
for each $m \in X$ such that the marginal $h(o\mid m,\lambda)$ is well-defined.
\end{definition}

\begin{theorem}\label{nogothm}\index{locality \& non-contextuality!observable properties}
An ontological model is local/non-contextual if and only if all measurements are of ontic observable properties.
\end{theorem}

\begin{proof}
First, we claim that a model is deterministic if and only if its observable properties are ontic. This holds since, by (\ref{prophv}),
\[
h(o\mid m,\lambda)=1 \qquad \Leftrightarrow \qquad f_m(\lambda) (o) = 1.
\]
Next, we claim that a model is parameter independent if and only if all measurements are of observable properties. This holds since, by definition \ref{def:obsprop}, all measurements are of observable properties if and only if all marginals $h(o\mid m,\lambda)$ are well-defined. The result follows.
\end{proof}

This is another characterisation of locality, which falls out easily from the definitions. It is similar to the Kochen-Specker \cite{kochen:75} or topos approach \cite{isham:98} treatment of non-contextuality. It can provide an alternative and sometimes simpler approach to certain results. The first result we mention shows that local ontological models have a canonical form. In fact, it shows that local ontological or hidden variable models can equivalently be expressed as distributions over the set of global assignments. (In this sense it shows how local ontological models are subsumed by the sheaf-theoretic approach.) It has recently been proved in measure theoretic generality \cite{brandenburger:13}, and can also be seen to generalise earlier work by Fine \cite{fine:82}. An interesting, related point is that, by allowing for negative probabilities, these canonical models can also generate all no-signalling correlations \cite{abramsky:11,abramsky:14,mansfield:13t}.

\begin{theorem}\label{thm:canonical}
Local models can be expressed in a \emph{canonical form}, with an ontic state space $\Omega := \mathcal{E}(X)$, and probabilities
\[
h(\overline{o}\mid \overline{m},\omega) = \prod_{m \in \overline{m}} \, \delta \left( \omega(m), \overline{o}(m) \right)
\]
for all $\overline{m} \in \mathcal{M}$, $\overline{o} \in \mathcal{E}(\overline{m})$, and $\omega \in \Omega$.
\end{theorem}

\begin{proof}
By theorem \ref{nogothm}, a local model $h$ over $\Lambda$ has a set $\{\widehat{f_m}:\Lambda \rightarrow O\}_{m \in X}$ of ontic observable properties. For each $\lambda \in \Lambda$, we define a function $\omega_\lambda \in \mathcal{E}(X)$ by $\omega_{\lambda}(m) := \widehat{f_m}(\lambda)$. Then the function $c:\Lambda \rightarrow \mathcal{E}(X)$ defined by $c(\lambda) := \omega_\lambda$ takes the original to the canonical ontic state space.

We first prove the claim that if $\lambda, \lambda' \in c^{-1}(\omega)$ for some $\omega \in \mathcal{E}(X)$ then
\[
h(\overline{o}\mid \overline{m},\lambda) = h(\overline{o}\mid \overline{m},\lambda')
\]
for all $\overline{m} \in \mathcal{M}$ and $\overline{o} \in \mathcal{E}(\overline{m}$). Since $\lambda, \lambda' \in c^{-1}(\omega)$, then $\omega_\lambda = \omega_{\lambda'}$, and therefore $\widehat{f_m}(\lambda) = \widehat{f_m}(\lambda')$ for all $m\in X$. It follows that
\begin{align*}
h(\overline{o}\mid \overline{m},\lambda) &= \prod_{m \in \overline{m}} \, h(\overline{o}(m)\mid m, \lambda) \\
&= \prod_{m \in \overline{m}} \, f_m(\lambda) \left( \overline{o}(m) \right) \\
&= \prod_{m \in \overline{m}} \, \delta \left( \widehat{f_m}(\lambda), \overline{o}(m) \right) \\
&= \prod_{m \in \overline{m}} \, \delta \left( \widehat{f_m}(\lambda'), \overline{o}(m) \right) \\
&= \cdots \\
&= h(\overline{o}\mid \overline{m},\lambda'),
\end{align*}
where the first equality can be shown to hold by locality.

The canonical model $h$ over $\Omega$ is defined by
\[
h(\overline{o}\mid \overline{m},\omega) := h(\overline{o}\mid \overline{m},\lambda_\omega)
\]
and
\[
h(\omega\mid p) := \sum_{\lambda \in c^{-1}(\omega)} \, h(\lambda\mid p)
\]
for all $\overline{m} \in \mathcal{M}$, $\overline{o} \in \mathcal{E}(X)$, $\omega \in \Omega$, $\lambda \in \Lambda$, and any $\lambda_\omega \in c^{-1}(\omega)$. The canonical model realises the same operational probabilities as the original, since
\begin{align*}
\sum_{\omega \in \Omega} \, & h(\overline{o}\mid \overline{m},\omega) \, h(\omega\mid p) \\ &= \sum_{\omega \in \Omega} \, \left(h(\overline{o}\mid \overline{m},\lambda_\omega) \,\sum_{\lambda \in c^{-1}(\omega)} \, h(\lambda\mid p)\right) \\
&= \sum_{\omega \in \Omega} \, \sum_{\lambda \in c^{-1}(\omega)} \, h(\overline{o}\mid \overline{m},\lambda) \, h(\lambda\mid p) \\
&= \sum_{\lambda \in \Lambda} \, h(\overline{o}\mid \overline{m},\lambda) \, h(\lambda\mid p),
\end{align*}
where the second equality holds by the previous claim. Moreover, the operational probabilities can be simplified as follows.
\begin{align*}
h(\overline{o}\mid \overline{m},\omega) &=  h(\overline{o}\mid \overline{m},\lambda_\omega) \\
&= \prod_{m \in \overline{m}} \, \delta \left( \widehat{f_m}(\lambda_\omega), \overline{o}(m) \right) \\
&= \prod_{m \in \overline{m}} \, \delta \left( \omega(m), \overline{o}(m) \right)
\end{align*}
\end{proof}

The next proposition will not be surprising in light of the EPR argument \cite{einstein:35}. It shows that if one were to take the view that quantum mechanics is $\psi$-complete then all non-trivial observables are epistemic or inherently probabilistic. Indeed, we can obtain a re-statement of the EPR result as a corollary.

\begin{proposition}\label{incomp}\index{$\psi$-completeness}
Any non-trivial quantum mechanical observable is epistemic with respect to $\psi$-complete quantum mechanics.
\end{proposition}

\begin{proof}
Any observable $\hat{A} \neq \mathbf{I}$ has eigenvectors, say $\left|  v_1 \right>$ and $\left|v_2 \right>$, corresponding to distinct eigenvalues, say $o_1$ and $o_2$. Consider any state $\left| \psi \right>$ such that $\left< v_1 | \psi \right> >0$ and $\left< v_2 | \psi \right> >0$. In a $\psi$-complete model, the wavefunction is the ontic state, so $\lambda = \left| \psi \right>$. Then
\[
f_{\hat{A}} (\lambda) (o_1) = h(o_1\mid \hat{A},\lambda) = \left| \left< v_1 | \psi  \right> \right|^2 >0,
\]
and similarly $f_{\hat{A}} (\lambda) (o_2) >0$. Therefore $f_{\hat{A}}$ is epistemic.
\end{proof}

\begin{corollary}[EPR]\label{prebell}\index{EPR argument}
Under the assumption of locality, quantum mechanics cannot be $\psi$-complete.
\end{corollary}

\begin{proof}
By proposition \ref{incomp}, any non-trivial quantum observable is epistemic with respect to $\psi$-complete quantum mechanics. Therefore, by theorem \ref{nogothm}, $\psi$-complete quantum mechanics is not local.
\end{proof}

This is the same result that was argued for by EPR, though this proof has more in common with an earlier argument by Einstein at the 1927 Solvay conference \cite{bacciagaluppi:10}, and also with a more recent, general treatment found in \cite{brandenburger:08} and in \cite{abramsky:10}.

\section{The PBR Theorem}\label{sec:pbr}

In this section we briefly make some observations relating to the PBR theorem, which deals with the reality (i.e. onticity in the sense of definitions \ref{def:hs} and \ref{def:ontepi}) of the wavefunction. One of the assumptions for this result is \emph{preparation independence} \cite{pusey:12}: \begin{quote} systems that are prepared independently have independent physical states.\end{quote} The other assumptions are implicit in the present framework.

\begin{theorem}[PBR]\label{pbrthm}
For any preparation independent theory that reproduces (a certain set of) quantum correlations, the wavefunction is ontic.
\end{theorem}

The preparation independence assumption is concerned with the composition of systems and has not appeared in previous no-go results. We will attempt to give this a more careful treatment. First of all, the PBR theorem describes a \emph{preparation scenario}\index{preparation scenario}. Generalising, this can be thought of as a kind of dual to a measurement scenario, in which the preparations $P$ play the role of measurements and the ontic states $\Lambda$ play the role of outcomes. Just as we had a compatibility structure $\M$ for measurements, which in Bell scenarios allowed us to chose one measurement from each site, we should in general have a compatibility structure $\mathcal{P}$ for preparations, which in the case of the PBR result allows us to chose one preparation from each site. We should allow for joint ontic states $\overline{\lambda}: \overline{p} \rightarrow \Lambda$, just as we allowed for joint outcomes. Similarly to before, we will take $\overline{\lambda}$ to denote a tuple of joint hidden variables, and $\overline{p}$ to denote a tuple of joint preparations, one for each site. The definitions of an ontological model and the properties from section \ref{ontsect} can be modified in the obvious way to account for this additional structure.

\begin{definition}
An ontological theory $h$ over $\Lambda$ is \emph{preparation independent} if and only if we can factor
\begin{equation}\label{eq:lsep}
h(\overline{\lambda} \mid \overline{p}) = \prod_{p \in \overline{p}} \, h( \lambda_p \mid \overline{p})
\end{equation}
for all $\overline{p} \in \mathcal{P}$, where $\lambda_p := \overline{\lambda}|_{p}$.
\end{definition}

Presented in this way, preparation independence (\ref{eq:lsep}) in a preparation scenario is clearly seen to be analogous to non-contextuality or Bell locality (\ref{eq:bellloc}) in a measurement scenario. An intriguing question is what happens if this is relaxed to an assumption analogous to no-signalling, in which we only assume that the marginal distributions $h(\lambda_p \mid \overline{p})$ are well-defined: a sort of `no-preparation-signalling' assumption.

\begin{definition}
An ontological theory $h$ over $\Lambda$ is \emph{no-preparation-signalling} if and only if the marginal probabilities $h(\lambda_p, \overline{p})$ are well-defined.
\end{definition}

In this case, it is easy to see that the PBR argument  of \cite{pusey:12} no longer holds. It is true that the relaxed assumption would allow for global or non-local correlations in the joint ontic state $\overline{\lambda}$; but perhaps, in light of the Bell and Kochen-Specker theorems, this should not be so surprising. An important question that remains to be answered, therefore, is whether by another argument a result similar to (or indeed counter to) that of PBR can be proved.

Another observation, which is also pointed out in \cite{harrigan:10}, is that onticity of the wavefunction is actually inconsistent with locality. This can be demonstrated as a consequence of what Schr\"{o}dinger called \emph{steering} \cite{schrodinger:36}. If a local measurement in the basis $\{\ket{0},\ket{1}\}$ is made on the first qubit of the state
\[ \ket{\phi^+} = \frac{1}{\sqrt{2}} \left( \ket{00} + \ket{11} \right) \]
then this can be considered as a remote preparation of the second qubit in one of the states $\ket{0}$ or $\ket{1}$, and similarly for a measurement in the basis $\{ \ket{+}, \ket{-} \}$. If the second sub-system has an ontic state $\lambda$ that is independent of measurements made elsewhere, then $\lambda$ must be consistent with one state from each of the sets $\{ \ket{0}, \ket{1} \}$ and $\{ \ket{+}, \ket{-} \}$, but this contradicts the onticity of the wavefunction.

We therefore arrive at the following theorem, which we propose to think of as a weak Bell theorem, since it draws the same conclusion as Bell's theorem \cite{bell:64} but with the extra assumption of preparation independence.

\begin{theorem}\label{thm:wkbell}
Quantum mechanics is not realisable by any preparation independent, local ontological theory.
\end{theorem}

\begin{proof}
This follows from the PBR theorem and the occurrence of steering in quantum mechanics.
\end{proof}

The ease at which this result falls out may lead us to be cautious of the strength of the preparation independence assumption.

\section{Discussion}

We have presented a generalisation and reformulation of the Harrigan-Spekkens criterion for the reality or onticity of the wavefunction. The reformulation can be thought of as a special case of the dual equivalence between the category of {von Neumann} algebras and $*$-homomorphisms and the category of measure spaces and measurable functions. It has been seen to have several advantages: it avoids measure theoretic technicalities relating to sets of measure zero and is mathematically and conceptually straightforward. Of course, it is also general enough to apply to any object or property in any ontological theory.

The first obvious application of the criterion to an object or property other than the wavefunction is to the observable properties of a system. This led to a new characterisation of locality and non-contextuality in terms of the nature of the observed properties. This can provide a useful tool for looking at foundational results: we have used it to obtain a short proof that local ontological models have a canonical form and to gain another perspective on the EPR argument. The characterisation is similar to the Kochen-Specker \cite{kochen:75} or topos approach \cite{isham:98}\index{topos approach} treatment of non-contextuality.

It is notable that the characterisation draws a connection between locality and onticity: these are the properties that are dealt with by the Bell and PBR theorems, respectively. A further connection was found in theorem \ref{thm:wkbell}, which showed that a weakened version of Bell's result can be obtained by an argument that combines the PBR result with the incompatibility of steering and the onticity of the wavefunction.

In relation to the PBR result itself, we have attempted to give a more careful treatment of the assumption of preparation independence, and made a concrete analogy between this property and locality/non-contextuality. It is possible to relax the assumption to something analogous to no-signalling, in which case we have pointed out that the PBR argument no longer holds. This amounts to introducing global or non-local correlations in the joint ontic state, which at least is consistent with the Bell and Kochen-Specker theorems. An open question is whether by another argument the result can be shown to hold with the relaxed assumption of `no-preparation-signalling'.

\end{multicols}

\bibliographystyle{abbrv}
\bibliography{refs}

\end{document}